\documentclass{article} 
\usepackage{amsmath, amssymb, amsthm}
\usepackage[dvipdfmx]{graphicx}
\begin{document}
\newtheorem{theo}{Theorem}
\newtheorem{lemm}{Lemma}
\newtheorem{defi}{Definition}
\newtheorem{prop}{Proposition}
\newtheorem{sket}{Skecth}
\title{$\mathsf{QMA}$(2) with postselection equals to $\mathsf{NEXP}$}
\author{Yusuke Kinoshita~\thanks{Graduate School of Informatics, Nagoya University,
kinoshita.yusuke@j.mbox.nagoya-u.ac.jp.}}
\maketitle

\begin{abstract}
We study the power of $\mathsf{QMA}(2)$ with postselection and show that the power is equal to $\mathsf{NEXP}$. Our method for showing this equality can be also used to prove that other classes with exponentially small completeness-soundness gap equals to the corresponding postselection versions.

\end{abstract}
\section{Introduction}
Verifying a proof sent from an unlimitedly powerful prover is one of the main issues in complexity theory. One of such complexity classes is $\mathsf{MA}$, introduced by Babai \cite{B}. This class is defined as the class of languages decided by a Merlin-Arthur system, that consists of a probabilistic polynomial time verifier called Arthur and  an infinitely powerful prover called Merlin. Arthur computes with a proof sent from Merlin, and each yes instances have at least one proof that Arthur accepts with high probability. This class is important in classical complexity theory.

In quantum complexity theory, there exists a similar class called Quantum Merlin-Arthur($\mathsf{QMA}$), and this has been intensively studied since introduced by Knill \cite{K}, Kitaev \cite{KSV}, and Watrous \cite{W}. In the most common setting, Merlin provides a quantum proof and Arthur is allowed to use polynomial time quantum computations.

Kobayashi et al. \cite{KMY} posed a question that a concatenation of many quantum proofs can be simulated by one quantum proof. In classical settings, it is obvious that many proofs can be simulated by a concatenated one proof, but one quantum proof may not be a  concatenation of non-entangled many proofs. Hence a direct simulation of the original proof system with concatenated one proof may be cheated by an entangled proof. This problem has a relation to the property of entanglement and fundamentals of quantum complexity, and it is natural to introduce the complexity class $\mathsf{QMA(2)}$, which is decided by a Merlin-Arthur system which uses two Merlins. $\mathsf{QMA(2)}$ has a natural complete problem arising from quantum chemistry, called Pure State N-Representability Problem \cite{LCV}. It is conjectured that witnesses without entanglements are hard to simulate by one witness, and $\mathsf{QMA}$ and $\mathsf{QMA}$(2) are different complexity classes. For example, Blier and Tapp \cite{BT} showed that $\mathsf{QMA}$(2) can solve SAT with $O(\log n)$ size witnesses and with polynomial inverse completeness-soundness gap, while such an algorithm is not known for $\mathsf{QMA}$. The gap parameter of this result was improved by Aaronson et al. \cite{ABDF}, who used $O(\sqrt n$poly$\log n)$ length witnesses and got a constant completeness-soundness gap.  Whether $\mathsf{QMA}(k)(k>2)$, which is an analogue of $\mathsf{QMA(2)}$ with $k$ proofs, is equal to $\mathsf{QMA(2)}$ was an important problem. This problem was resolved by Harrow and Montanaro \cite{HM}. They showed $\mathsf{QMA}(2)=\mathsf{QMA}(k)$(for any $k>2$).

In quantum complexity theory, complexity classes with postselection has been used to show quantum supremacy of sub-universal models such as Boson Sampling \cite{AA}, IQP \cite{BJS,BMS}, and DQC1 \cite{FKM+,MFF}. These results mean that the difficulty of the simulation of probabilistic distribution of quantum circuits relates to separation of Polynomial Hierarchy. Recently postselection  beyond $\mathsf{BQP}$ was investigated in \cite{MN,UHB}. Morimae and Nishimura \cite{MN} showed $\mathsf{postQMA}=\mathsf{PSPACE}$, where $\mathsf{postQMA}$ is the postselection version of $\mathsf{QMA}$, using the result that $\mathsf{PSPACE}$ is equal to $\mathsf{QMA}$ with exponentially small gap \cite{FL}. They also showed several results on complexity classes with postseletion.

In this article, we define $\mathsf{postQMA}$(2) and show $\mathsf{postQMA}(2)=\mathsf{NEXP}$. The class $\mathsf{postQMA}(k)$ is easily computed by $\mathsf{NEXP}$, and hence the power of quantum computation with non-entangled witnesses and postselection is characterized exactly. Here we describe the main technical difficulty briefly, and more details are in the appendix. The previous techniques \cite{Aa, MN} is preparing $\alpha|0\rangle+\beta|1\rangle$ by using output qubit and detecting $\beta\lessgtr 0$. It is necessary to erase the garbage, that is, natural quantum computing makes a state in the form of  $\alpha|0\rangle|\phi_0\rangle+\beta|1\rangle|\phi_1\rangle$, but what we needs  is $\alpha|0\rangle+\beta|1\rangle$. To erase the garbage, Aaronson \cite{Aa} uses reversible computation of classical circuit and applying Hadamard transformation to computational basis, and Morimae and Nishimura \cite{MN} use distillation \cite{KGN}. Both techniques cannot be applied to $\mathsf{postQMA}(2)$. The garbage of general quantum computation cannot be erased in contrast to superposition of  computations of classical reversible circuit, and the witnesses cannot be restricted to eigenvectors since the witness space is not linear, therefore distillation cannot be applied to $\mathsf{postQMA}(2)$. These difficulties are linked to the difficulty of computation with non-entangled proofs, and it seems difficult to analyze $\mathsf{postQMA}(2)$ by the previous protocol. Our technique is restricting completeness to 1 or bounding completeness error, and bounding the acceptance probability even with garbage superposition. Our technique is useful for other complexity classes to prove that the corresponding classes with exponentially small completeness/soundness gaps are equal to the postselection versions. For example we define $\mathsf{postQIP}$ and sketch the proof that $\mathsf{postQIP}$ equals to $\mathsf{QIP}$ with exponentially small gap. 

The remainder of the paper is organized as follows.  Section 2 defines a new complexity class $\mathsf{postQMA(2)}$. Section 3 shows the main result. In section 4 we show another result about postseletion to compare with the main result and discuss about the case that completeness is less than $1$. In section 5 we give several open problems.

\section{Preliminaries}
We assume that readers are familiar with quantum computation \cite{KSV} and classical computational complexity \cite{AB}. In this section we define $\mathsf{QMA}(2)(c,s)$ and $\mathsf{postQMA}(2)(c,s)$.
\begin{defi}[$\mathsf{QMA(2)}(c,s)$]
Let $L$ be a language. Let $c(n),s(n): \mathbb{Z}\rightarrow[0,1]$ be functions that can be computed in polynomial time and $c(n)-s(n)$ is positive and larger than the inverse of some polynomial if $n$ is sufficiently large. $L$ is in $\mathsf{QMA(2)}(c,s)$ if there exist polynomials $w(n)$, $m(n)$, and a uniform quantum circuit family $\{V_x\}$ constructed in polynomial time that satisfies follows for any $n$ and any string $x$ with $|x|=n$:
\begin{item}
\item
if $x\in L$, then there exist 2 $w(n)$-qubit states $|\psi_1\rangle$ and $|\psi_2\rangle$ such that
\begin{equation*}
{\rm Pr}_{V_x(|\psi_1\rangle|\psi_2\rangle)}[o=1]\ge c,
\end{equation*}
\item
if $x\notin L$, then for any 2 $w(n)$-qubit states $|\psi_1\rangle$ and $|\psi_2\rangle$,
\begin{equation*}
{\rm Pr}_{V_x(|\psi_1\rangle|\psi_2\rangle)}[o=1]\le s.
\end{equation*}
\end{item}
Here, ${\rm Pr}_{V_x(|\psi_1\rangle|\psi_2\rangle)}[o=1]$ is the probability that $V_x$ with inputs $|\psi_1\rangle|\psi_2\rangle$ outputs $o=1$. Namely it is defined by 
\begin{equation*}
\begin{split}
&{\rm Pr}_{V_x(|\psi_1\rangle|\psi_2\rangle)}[o=1] =\\
&\ \ {\rm Tr}[|1\rangle\langle1|\otimes I^{\otimes 2w(n)+m(n) -1}Q_x(|\psi_1\rangle\langle\psi_1|\otimes |\psi_2\rangle\langle\psi_2|\otimes |0\rangle\langle0|^{\otimes m(n)}){Q_x}^\dagger].
\end{split}
\end{equation*}
\end{defi}

\begin{defi}[$\mathsf{postQMA(2)}(c,s)$]
Let $L$ be a language. Let $c(n),s(n): \mathbb{Z}\rightarrow[0,1]$ be functions that can be computed in polynomial time and $c(n)-s(n)$ is positive and larger than the inverse of some polynomial if $n$ is sufficiently large. $L$ is in $\mathsf{postQMA(2)}(c,s)$ if there exist polynomials $w(n)$, $m(n)$, $l(n)$, and a uniform quantum circuit family $\{V_x\}$ constructed in polynomial time that satisfies follows for any $n$ and any string $x$ with $|x|=n$:
\begin{item}
\item For all 2 $w(n)$-qubit states $|\psi_1\rangle$ and $|\psi_2\rangle$,
\begin{equation*}
{\rm Pr}_{V_x(|\psi_1\rangle|\psi_2\rangle)}[p=1]\ge 2^{-l(n)},
\end{equation*}
\item
if $x\in L$, then there exist 2 $w(n)$-qubit states $|\psi_1\rangle$ and $|\psi_2\rangle$ such that 
\begin{equation*}
 {\rm Pr}_{V_x(|\psi_1\rangle|\psi_2\rangle)}[o=1|p=1]\ge c(n),
\end{equation*}
\item
if $x\notin L$, then for any 2 $w(n)$-qubit states $|\psi_1\rangle$ and $|\psi_2\rangle$,
\begin{equation*}
{\rm Pr}_{V_x(|\psi_1\rangle|\psi_2\rangle)}[o=1|p=1]\le s(n).
\end{equation*}
\end{item}
Here, ${\rm Pr}_{V_x(|\psi_1\rangle|\psi_2\rangle)}[o=1|p=1]$ is the conditional probability that $V_x$ with inputs $|\psi_1\rangle|\psi_2\rangle$ outputs $o=1$ with the condition $p=1$. This probability is caluculated as follows.
\begin{equation*}
\begin{split}
{\rm Pr}_{V_x(|\psi_1\rangle|\psi_2\rangle)}[o=1|p=1]=\frac{{\rm Pr}_{V_x(|\psi_1\rangle|\psi_2\rangle)}[o=1,p=1]}{{\rm Pr}_{V_x(|\psi_1\rangle|\psi_2\rangle)}[p=1]}.
\end{split}
\end{equation*}
Each term in LHS is defined by
\begin{equation*}
\begin{split}
&{\rm Pr}_{V_x(|\psi_1\rangle|\psi_2\rangle)}[o=1,p=1]=\\
&\ \ {\rm Tr}[|1\rangle\langle1|\otimes |1\rangle\langle1|\otimes I^{\otimes 2w(n)+m(n) -1}Q_x|\psi_1\rangle\langle\psi_1|\otimes |\psi_2\rangle\langle\psi_2|\otimes|0\rangle\langle0|^{\otimes m(n)}{Q_x}^\dagger],\\
&{\rm Pr}_{V_x(|\psi_1\rangle|\psi_2\rangle)}[p=1]=\\
&\ \ {\rm Tr} [I\otimes |1\rangle\langle1|\otimes I^{\otimes 2w(n)+m(n) -1}Q_x(|\psi_1\rangle\langle\psi_1|\otimes |\psi_2\rangle\langle\psi_2|\otimes|0\rangle\langle0|^{\otimes m(n)}){Q_x}^\dagger].
\end{split}
\end{equation*} 
\end{defi}
\section{Main Result}
\begin{theo}
There exists a constant $s$ such that $\mathsf{postQMA}(2)(1,s)=\mathsf{NEXP}$.
\end{theo}
Theorem 1 is sufficient to prove that there is a polynomial $p$ which satisfies $\mathsf{postQMA}(2)(1,2^{-p})=\mathsf{NEXP}$, since we can amplify the gap between completeness and soundness by repetition of witnesses. This is also sufficient to prove that $\mathsf{postQMA}(2)(c,s)=\mathsf{NEXP}$ for any constant $c<1$, since it is enough to prepare a new 1 qubit $\sqrt c|0\rangle+\sqrt{1-c}|1\rangle$ and accept if the original completeness 1 protocol accepts and the new qubit outputs 1.
First, we prove {$\mathsf{postQMA(2)}\subseteq \mathsf{NEXP}$.
\begin{prop}
{$\mathsf{postQMA(2)}\subseteq \mathsf{NEXP}$.}
\end{prop}
\begin{proof}
The witnesses of $\mathsf{postQMA(2)}$ are poly-length qubits, and hence exponential length classical bits can describe these witnesses with exponential precision, and the acceptance probability can be computed in non-deterministic exponential time.
\end{proof}
Next, we prove that $\mathsf{postQMA(2)}\supseteq \mathsf{NEXP}$. We use the next lemma from previous results.
\begin{lemm}[\cite{BT,P}]
For any polynomial $r$, $\mathsf{NEXP} \subseteq \mathsf{QMA}(2)(1, 1-2^{-r})$ holds.
\end{lemm}
The sketch of this theorem is as follows: 3COLOR can be solved with log size witnesses, if the comleteness/soundness gap is the inverse of a polynomial \cite{BT}. A similar proof is efficient for succinct 3COLOR \cite{P}.

The next lemma is our main technical result.
\begin{lemm}
There exists a constant $s$ such that $\mathsf{QMA}(2)(1,1-2^{-r})\subseteq$\\$ \mathsf{postQMA}(2)(1,s)$ holds for any polynomial $r$.
\end{lemm}
\begin{figure}
\hrulefill\\
Protocol 1\\\\
$Q_x$ is the circuit for $L\in \mathsf{QMA(2)}(1, 1-2^{-r})$. 
Denote $|\psi\rangle=|\psi_{x,1}\rangle|\psi_{x,2}\rangle|0^n\rangle$: which is witnesses and ancillas.\\
We first prepare $|0\rangle|0\rangle|\psi\rangle$.\\
1. Apply $Q_x$  to $|\psi\rangle$.\\
2. Copy the output qubit to the second qubit.\\
3. Apply $Q_x^{-1}$. \\
4. Prepare $\epsilon|1\rangle-|0\rangle$ in the first qubit.\\
5. Apply the unitary operator $\frac{1}{1+\epsilon^2}\left(
\begin{array}{rrr}
      1 & \epsilon \\
      -\epsilon &1\\
    \end{array}
\right)$ to the second qubit. \\
6. Measure the first and second qubits by projection onto $\{|00\rangle,|11\rangle\}/$\\$\{|01\rangle,|10\rangle\}$ and postselect $\{|00\rangle,|11\rangle\}$.\\
7. Measure the first and second qubits by projection onto $\{|00\rangle+|11\rangle\}/\{|00\rangle-|11\rangle, |01\rangle+|10\rangle, |01\rangle-|10\rangle\}$.\\
\hrulefill
\caption{Protocol of $\mathsf{postQMA(2)}$ transfered from $\mathsf{QMA}(2)(1, 1-2^{-r})$. While step 3 is not necessary, but we add it for our analysis.}
\label{fig.1}
\end{figure}
\begin{proof}
Suppose a circuit family $\{Q_x\}$ solves $L\in \mathsf{QMA}(2)(1,1-2^{-r})$. We construct a circuit family of post$\mathsf{QMA}(2)$ from $\{Q_x\}$. The protocol is in Figure \ref{fig.1}.
The analysis is similar to \cite{MN}, using Distillation \cite{KGN}, but the main difference is that  there remains $|\bot\rangle$ orthogonal to the input state, since witnesses that maximize the acceptance probability may not be an eigenstate if $x$ is a no instance.

Suppose $|\psi_{x,1}\rangle|\psi_{x,2}\rangle$ is the witness that maximize the acceptance probability of $Q_x$. We use 2 qubits in addition to $|\psi_{x,1}\rangle|\psi_{x,2}\rangle|0^n\rangle$. There exist some states $|\phi_{x,0}\rangle$,$|\phi_{x,1}\rangle$ that satisfy the following equation 
\begin{equation}\label{eq:eq1}
Q_x|\psi_{x1}\rangle|\psi_{x2}\rangle|0^n\rangle=\sqrt{p_x}|1\rangle|\phi_{x,1}\rangle+\sqrt{1-p_x}|0\rangle|\phi_{x,0}\rangle 
\end{equation}
(\ref{eq:eq1}) is the state after step 1 in Figure 1. Prepare 1 qubit $|0\rangle$ in addition to (\ref{eq:eq1}). Apply CNOT on the new qubit and
the first qubit of (\ref{eq:eq1}) in step 2, where the latter is the control qubit. After step 2 the state is (\ref{eq:eq2}).
\begin{equation}\label{eq:eq2}
\sqrt{p_x}|11\rangle|\phi_{x,1}\rangle+\sqrt{1-p_x}|00\rangle|\phi_{x,0}\rangle. 
\end{equation}
Apply $Q_x^{-1}$ on the qubits of (\ref{eq:eq2}) that originally $Q_x$ acts on. The next state is as follows.
\begin{equation}\label{eq:eq3}
Q_x^{-1}\sqrt{p_x}|11\rangle|\phi_{x,1}\rangle+Q_x^{-1}\sqrt{1-p_x}|00\rangle|\phi_{x,0}\rangle=\sqrt{p_x}|1\rangle|f_1\rangle+\sqrt{1-p_x}|0\rangle |f_0\rangle.
\end{equation}
Denote $|\psi_{x,1}\rangle|\psi_{x,2}\rangle|0^n\rangle=|\psi\rangle$ and $\Pi_1=|1\rangle\langle1|\otimes I^{\otimes w(n)+m(n) -1}$.  Since 
$\langle\psi|\sqrt{p_x}|f_1\rangle=\langle\psi|{Q_x}^\dagger\Pi_1Q_x|\psi\rangle=|\Pi_1Q_x|\psi\rangle|^2=p_x$, $\langle\psi|f_1\rangle=\sqrt{p_x}$ and there exists $|\bot\rangle$ orthogonal to $|\psi\rangle$ such that
$|f_1\rangle=\sqrt{p_x}|\psi\rangle+\sqrt{1-p_x}|\bot\rangle$. Similary $|f_0\rangle$ can be written by $|\bot _2\rangle$ orthogonal to $|\psi\rangle$ as follows : 
$|f_0\rangle=\sqrt{1-p_x} |\psi\rangle+\sqrt{p_x}|\bot_2\rangle$.
Note that $\sqrt{p_x}|f_1\rangle+\sqrt{1-p_x}|f_0\rangle= |\psi\rangle$, $|\bot_2\rangle=-|\bot\rangle$, and $|f_0\rangle=\sqrt{1-p_x} |\psi\rangle-\sqrt{p_x}|\bot\rangle$.\\
RHS of (\ref{eq:eq3}) equals to
\begin{equation}(p_x|1\rangle+(1-p_x)|0\rangle)|\psi\rangle+\sqrt{p_x(1-p_x)}(|1\rangle-|0\rangle)|\bot\rangle.
\end{equation}Define $\epsilon=2^{-10r}$. Prepare another new 1 qubit in state $\epsilon |1\rangle-|0\rangle$ (we omit the normalization factor $1/(1+\epsilon^2)$ for convenience). The state of whole qubits are as follows.
\begin{equation}
(\epsilon |1\rangle-|0\rangle) \{(p_x|1\rangle+(1-p_x)|0\rangle)|\psi\rangle+\sqrt{p_x(1-p_x)} (|1\rangle-|0\rangle)|\bot\rangle\}.
\end{equation}

Apply a unitary operator :$|0\rangle\rightarrow |0\rangle+\epsilon|1\rangle, |1\rangle\rightarrow|1\rangle-\epsilon|0\rangle$ on the second qubit iagain we omit the normalization factor). The whole state will become as follows: 
\begin{equation}
\begin{split}
(\epsilon |1\rangle - |0\rangle)\{&(p_x+(1-p_x)\epsilon)|1\rangle+((1-p_x)-\epsilon p_x)|0\rangle)|\psi\rangle\\
&+\sqrt{p_x(1-p_x)} ((1-\epsilon)|1\rangle-(1+\epsilon)|0\rangle)|\bot\rangle\}.
\end{split}
\end{equation}
Measure the first 2 qubits by projection to the space spaned by $\{|00\rangle,|11\rangle\}$/$\{|01\rangle,$\\$|10\rangle\}$, and postselect $\{|00\rangle,|11\rangle\}$. The states becomes:
\begin{equation}
\label{eq:eq7}
\begin{split}
\{\epsilon(p_x+(1-p_x)\epsilon)|11\rangle-((1-p_x)-\epsilon p_x)|00\rangle\}|\psi\rangle\\
+\sqrt{p_x(1-p_x)} \{\epsilon(1-\epsilon)|11\rangle+(1+\epsilon)|00\rangle\}|\bot\rangle.
\end{split}
\end{equation}
Measure the first 2 qubit by projection to $\{|00\rangle+|11\rangle\}$ and the complement, and accept if $|00\rangle+|11\rangle$ is measured.

Now we analyze the acceptance probability. If the instance is yes, then $p_x=1$ and hence $|00\rangle+|11\rangle$ is measured with probability 1. Assume the instance is no. Hereinafter $|\mu\rangle$, $|\mu'\rangle$ mean vectors with $O(\epsilon)$ norms. (7) can be denoted as follows.
\begin{equation}
\begin{split}
&\{-(1-p_x)|00\rangle+|\mu\rangle\}|\psi\rangle\\
&\ +\sqrt{p_x(1-p_x)} \{|00\rangle+|\mu'\rangle\}|\bot\rangle.
\end{split}
\end{equation}

Since $1-p_x\ge2^{-r}$ and $\epsilon=2^{-10r}$, the probability of projection to $\{|00\rangle+|11\rangle\}$ of this vector is $\frac{1}{\sqrt2}+O(\epsilon/(1-p_x))$.
\end{proof}
Theorem 1 is proved from Lemmas 1 and 2.

\section{Applications to other complexity classes and to the case that completeness is strictly smaller than 1}
\subsection{postQIP}
In this subsection, to show an example of application of our techniques to other classes, we define $\mathsf{postQIP}$ and prove $\mathsf{postQIP}$ equals to $\mathsf{QIP}$ with exponentially small gap. It can be proved that most classes with completeness one and with postselection are equal to themselves with exponentially small gap by our technique. We omit details since the technique is very similar to the main theorem. We remark that this relation is not obvious by techniques in previous research \cite{Aa,MN}.

\begin{defi}
$(\mathsf{postQIP})$\\
A language $L$ is in $\mathsf{postQIP}$ if
there exist a polynomial $r(n)$ and a verifier $\{V_x\}$ who can do quantum polynomial time computation and outputs 2 bit $o$ and $p$ such that;
\begin{item}
\item For any prover $P$, 
\begin{equation*}{\rm Pr}[\langle V, P\rangle(p=1)]\ge1/2^{r(|x|)},\end{equation*}
\item if $x\in L$, then there exists a prover $P$ such that \begin{equation*}{\rm Pr}[\langle V,P\rangle (o=1|p=1)] \ge 2/3,\end{equation*}
\item if $x\notin L$, then  for any prover $P$, \begin{equation*}{\rm Pr}[\langle V,P\rangle (o=1|p=1)]\le 1/3.\end{equation*}
\end{item}
Here, ${\rm Pr}[\langle V, P\rangle(p=1)]$ means the probability that $V$ outputs $p=1$ at the end of the interactive protocol between $V$ and $P$,  and ${\rm Pr}[\langle V, P\rangle(o=1|p=1)]$ means the conditional probability that $V$ outputs $o=1$ at the end of the interactive protocol between $V$ and $P$ with $p=1$.
\end{defi}
\begin{defi}
$(\mathsf{QIP_{exp}})$\\
A language $L$ is in $\mathsf{QIP_{exp}}$ if
there exist a verifier $V$ and functions of $|x|$, $c(|x|),s(|x|)$ such that $c(|x|)-s(|x|)>\frac{1}{exp},\ 0\le c(|x|),s(|x|)\le1$ satisfying followings:
\begin{item}
\item if $x\in L$, then there exists a prover $P$ such that \begin{equation*}{\rm Pr}[\langle V,P\rangle(o=1)] \ge c(|x|),\end{equation*}
\item if $x\notin L$, then for any prover $P$, \begin{equation*}{\rm Pr}[\langle V,P\rangle (o=1|p=1)]\le s(|x|).\end{equation*}
\end{item}
\end{defi}
\begin{prop}
$\mathsf{postQIP=QIP_{exp}}$
\end{prop}
$(Sketch)$. First we show $\mathsf{{QIP}_{exp}}$ can be computed with completeness 1 by similar techniques to prove usual $\mathsf{QIP}$ can be computed with completenes 1\cite{KW}. Next we use protocol 2, which is almost the same as protocol 1, except for the first state in the protocol 1,  protocol 2 uses the final state of the interactive protocol just before measuring, instead of the witnesses $|\psi\rangle$ of $\mathsf{postQMA}(2)$
\begin{figure}
\hrulefill\\
Protocol 2\\\\
Let $Q_x$ be the last verifier's circuit for $L\in \mathsf{QIP}(1, 1-2^{-r})$. Let $|\psi\rangle$ be the state just before applying $Q_x$:\\\\ 
1. Apply $Q_x$  to $|\psi\rangle$.\\
2. Copy the output qubit to the second qubit.\\
3. Apply $Q_x^{-1}$. \\
4. Prepare $\epsilon|1\rangle-|0\rangle$ in the first qubit\\
5. Apply the unitary operator $\frac{1}{1+\epsilon^2}\left(
\begin{array}{rrr}
      1 & \epsilon \\
      -\epsilon &1\\
    \end{array}
\right)$ to the second qubit. \\
6. Measure the first and second qubits by  projection onto $\{|00\rangle,|11\rangle\}/$\\$\{|01\rangle,|10\rangle\}$ and postselect $\{|00\rangle,11\rangle\}$.\\
7. Measure the first and second qubits by projection onto $\{|00\rangle+|11\rangle\}/\{|00\rangle-|11\rangle, |01\rangle+|10\rangle, |01\rangle-|10\rangle\}$.\\
\hrulefill
\caption{Protocol of $\mathsf{postQIP}$ transfered from $\mathsf{QIP}(1, 1-2^{-r})$. Though this is almost same to $\mathsf{postQMA(2)}$, we include this to show an application of our technique to other classes explicitly.}
\label{fig.2}
\end{figure}
\subsection{The case that completeness is strictly less than 1}
In this subsection we amplify the completeness/soundness gap of a protocol of which completeness is strictly smaller than 1 by postselection. We assume that $1\gg\delta^2\gg\epsilon$. 
The statement is as follows.
\begin{prop}
For any $\mathsf{QIP_{exp}}$ protocol with completeness $1-\epsilon$ and soundness $1-\delta$ $(\epsilon\ll\sqrt\delta)$, there exists a $\mathsf{postQIP}$ protocol with completeness/soundness gap larger than the constant.
\end{prop}
Proposition 2 is enough to prove that $\mathsf{QIP_{exp}}$ with 3 or more rounds is equal to $\mathsf{postQIP}$, since 3 rounds interactive protocols can be transformed to protocols with completeness 1 (normal QIP case is in \cite{KW}, and almost same to the exponential gap case.). On the other hand, proposition 3 is also true for 2 rounds interactive proof, but it is not clear how to make completeness near to 1 for 2 round protocols.  
\begin{figure}
\hrulefill\\
Protocol 3\\\\
Let $Q_x$ be the last verifier's circuit for $L\in \mathsf{QIP}(1-\epsilon, 1-\delta)$ and $Q_x$ accepts no instances with probability at least $1-\sqrt\delta$. Let $|\psi\rangle$ be the state just before applying $Q_x$.\\\\ 
1. Apply $Q_x$  to $|\psi\rangle$.\\
2. Copy the output qubit to the second qubit.\\
3. Apply $Q_x^{-1}$. \\
4. Prepare $\delta|1\rangle-|0\rangle$ in the first qubit.\\
5. Apply the unitary operator $\frac{1}{1+\delta^2}\left(
\begin{array}{rrr}
      1 & \delta \\
      -\delta &1\\
    \end{array}
\right)$ to the second qubit. \\
6. Measure the first and second qubits by  projection onto $\{|00\rangle,|11\rangle\}/$\\$\{|01\rangle,|10\rangle\}$ and postselect $\{|00\rangle,|11\rangle\}$.\\
7. Measure the first and second qubits by projection onto $\{|00\rangle+|11\rangle\}/\{|00\rangle-|11\rangle, |01\rangle+|10\rangle, |01\rangle-|10\rangle\}$.\\
\hrulefill
\caption{Protocol of $\mathsf{postQIP}$ transfered from $\mathsf{QIP}(1-\epsilon, 1-\delta)$. The differences between protocol 1 and protocol 2 are the rotations in step 4 and 5.}
\label{fig.3}
\end{figure}
\begin{proof}
The protocol is in Figure 3. First, we analyze the acceptance probability $p_x$ of yes instances. Let $\epsilon'=1-p_x$. After the step 4, the state is as follows.
\begin{equation}
\begin{split}
(\delta|1\rangle-|0\rangle) \{&((1-\epsilon')|1\rangle+\epsilon'|0\rangle)|\psi\rangle\\&+\sqrt{\epsilon'(1-\epsilon')} (|1\rangle-|0\rangle)|\bot\rangle\}.
\end{split}
\end{equation}
In the next equations, $|\mu\rangle$ denotes a vector with $O(\epsilon')$ norm, and $|\tau\rangle$ denotes a vector with $O(\delta)$ norm.
After the rotation in step 5 and postselection in step 6, we have the following state.
\begin{equation}
\begin{split}
&(\delta(|11\rangle+|00\rangle)+|\mu\rangle)|\psi\rangle\\&+\sqrt{\epsilon'(1-\epsilon')} (\delta|11\rangle-|00\rangle+|\tau\rangle)|\bot\rangle.
\end{split}
\end{equation}
Since $\sqrt{\epsilon'}\le\sqrt\epsilon\ll\delta$, if we measure this state in $\{|00\rangle+|11\rangle\}$ and its orthogonal vectors, we accept with probability $1-O(\epsilon/\delta^2)$

Next, we analyze the acceptance probability of no instances. Let $\delta'=1-p_x$. After the step 4, the state is as follows.
\begin{equation}
\begin{split}
(\delta|1\rangle-|0\rangle) \{&((1-\delta')|1\rangle+\delta'|0\rangle)|\psi\rangle\\&+\sqrt{\delta'(1-\delta')} (|1\rangle-|0\rangle)|\bot\rangle\}.
\end{split}
\end{equation}
In the next equations, $|\tau\rangle$ denotes a vector with $O(\delta\delta')$ norm, and $|\tau'\rangle$ denotes a vector with $O(\delta)$ norm.
After the rotation (step 5) and postselection (step 6), the state is as follows.
\begin{equation}
\begin{split}
&(\delta|11\rangle+(\delta+\delta')|00\rangle)+|\tau\rangle)|\psi\rangle\\&+\sqrt{\delta'(1-\epsilon')} (\delta|11\rangle-|00\rangle+|\tau'\rangle )|\bot\rangle.
\end{split}
\end{equation}
Since $\delta\le\delta'\ll 1$, if we measure this state in $\{|00\rangle+|11\rangle\}$ and orthogonal vectors, we accept with probability $1-\Omega(1)$.
\end{proof}

\section{Open Problems}
We conclude the paper by posing the following four questions. 

\begin{itemize}
\item The upper bounds of extremely small gap $\mathsf{QMA}$(2):

Though double-exponential gap $\mathsf{QMA}$(2) is contained in $\mathsf{NEXP}$, the upper bounds of $\mathsf{QMA}$(2) with infinitely small gap is not obvious. 
About $\mathsf{QMA}$, $\mathsf{QMA}$ with infinitely small gap is bounded by $\mathsf{ EXPSPACE}$. Moreover, if the gates are represented by algebraic numbers, and the completeness is 1, then the corresponding class is bounded by $\mathsf{PSPACE}$\cite{IKW}.
\item Decide whether small gap has more power or not:

Exponentially small gap will be more powerful than polynomial gap and some classes have strong evidence. Following examples have relatively strong evidences; 
\begin{itemize}
\item $\mathsf{QMA}$: $\mathsf{QMA}$ is in $\mathsf{PP}$, but $\mathsf{QMA}$ with exponentially small gap contains $\mathsf{PSPACE}$.
\item $\mathsf{BQP}$: $\mathsf{BQP}$ will not be able to compute $\mathsf{NP}$, but  $\mathsf{BQP}$ with exponentially small gap can compute $\mathsf{PP}$.
\item  $\mathsf{QMIP^*}$: $\mathsf{QMIP^*}$ is an example that the corresponding class with exponential small gap is more powerful, unless  $\mathsf{QMIP^*}$ equals to the exponential time version, since \cite{Ji} showed that  $\mathsf{QMIP^*}$ with exponentially small gap coincides with the exponential time version of  $\mathsf{QMIP^*}$. 
\end{itemize}

But  $\mathsf{QIP}$ seems not to have such evidence. The problem that exponentially small gap has more power than  polynomial gap or not remains open.

\ Another problem is to prove that some complexity class with exponential gap is strictly powerful than polynomial gap without any assumptions. The above examples need some computational assumptions, such as $\mathsf{PP\neq PSPACE}$, $\mathsf{NP\nsubseteq BQP}$ and that $\mathsf{QMIP^*}$ is strictly less powerful than exponential time of it. These assumptions will be extremely difficult to prove. To the best of our knowledge, there are no complexity classes of which exponentially small gap version has strictly stronger than polynomial gap version that we can prove without any assumptions.
\item Direct amplification of exponentially small gap for quantum complexity classes with completeness strictly smaller than 1:

Our proof needs some assumption of completeness and soundness. If completeness is strictly smaller than 1, the state $|\bot\rangle$ remains even in computations of yes instances, and our protocol fails. The postselection technique for more general classes remains open.
\end{itemize}

\section{Acknowledgments}
We thank to Prof. Kazuhisa Makino for supervising the author Kinoshita.
We thank to Dr. Francois Le Gall, from Kyoto Univ, Dr. Harumichi Nishimura, from Nagoya.Univ and Dr. Tomoyuki Morimae from Kyoto Univ. for helpful discussions.
\section{Appendix}

\subsection*{Overview of $\mathsf{PP=postBQP}$ \cite{Aa} and $\mathsf{postQMA=PSPACE}$ \cite{MN}}
Here we briefly describe the studies of postselection of Aaronson \cite{Aa} and Morimae and Nishimura \cite{MN} and why their protocol cannot be directly applied to $\mathsf{postQMA(2)}$.

Aaronson's protocol to prove $\mathsf{PP=postBQP}$ is as follows.\\\\
1. Make $\frac{1}{\sqrt N}\sum_{r}|r\rangle$, where $N$ is a normalization factor. Here, $r$ corresponds to a random string of the original PP algorithm.\\
2. Compute $\frac{1}{\sqrt N}\sum_{r} |r\rangle|b_r\rangle|garbage_r\rangle$, where $b_r$ is the output of the original PP algorithm which uses $r$ as a random number, and $|garbage_r\rangle$ is the garbage that depends on $r$.\\
3. Erase $|garbage_r\rangle$ by classical reversible circuit computation.\\
4. Apply Hadamard gates to $|r\rangle$, and then the state becomes $\frac{1}{\sqrt N}\sum_{r,y} (-1)^{r\cdot y}|y\rangle|b_r\rangle$.\\
5. Measure the register $|y\rangle$ and postselect $|0^n\rangle$. The state becomes $ p_{accept}|0\rangle+(1-p_{accept})|1\rangle$.\\
6. Prepare a new 1 qubit $\alpha|0\rangle+\beta|1\rangle$, for some $\alpha,\beta$, apply the controlled-Hadamard gate on $(\alpha|0\rangle+\beta|1\rangle)(p_{accept}|0\rangle+(1-p_{accept})|1\rangle)$, where $\alpha|0\rangle+\beta|1\rangle$ is the control qubit. \\
7. Measure the non-control qubits in the computational basis and postselect $|1\rangle$. \\
8. Measure the control-qubit in the computatioal basis.\\

It is critical in step 3, 4, and 5. that the computation is the linear sum of classical reversible computation. Hence it is difficult to apply this protocol to $\mathsf{postQMA(2)}$\\

Morimae and Nishimura's protocol \cite{MN} to prove $\mathsf{postQMA=QMA_{exp}(=PSPACE \cite{FL})}$ is as follows.\\\\
0. Witness is the eigenvector $|\psi\rangle$ of $\Pi_{0^n}Q_x^\dagger\Pi_{acc}Q_x$.\\
1. Apply $Q_x$.\\
2. Prepare a new 1 qubit $|0\rangle$ and apply CNOT on the deciding qubit and the new qubit, where the control qubit is the deciding qubit.\\
3. Apply $Q_x^\dagger$ and measure ancillas by projection onto $\{\Pi_{0^n}, I-\Pi_{0^n}\}$, and postselect $\Pi_{0^n}$ The new qubit in step 2. becomes $ p_{accept}|0\rangle+(1-p_{accept})|1\rangle$. \\
4-6. Similar to steps 6--8 of Aaronson's protocol.\\

Preparing the eigenvector is critical. Otherwise, the 1 qubit after step 3 entangles to remaining qubits and the 1 qubit is in mixed state. Postselection on mixed states will output a useless state, like in our protocol for no instances.

\end{document}